\documentclass{article}
\usepackage[pdftex]{graphicx}
\usepackage{amsmath}
\usepackage{amssymb}
\usepackage{amsthm}
\newtheorem{theorem}{Theorem}
\theoremstyle{remark}
\newtheorem{remark}{Remark}
\usepackage{array}

\begin{document}

\title{Electrostatic Discharge Currents Representation using the Multi-Peaked Analytically Extended Function by Interpolation on a D-Optimal Design}

\author{
Karl Lundeng\aa{}rd$^{*}$, Milica Ran\v{c}i\'{c}$^{*}$, \\
Vesna Javor$^{\dagger}$ and 
Sergei Silvestrov$^{*}$
\\ \\
$^{*}$Division of Applied Mathematics, \\
UKK, M\"{a}lardalen University\\
H\"{o}gskoleplan 1, Box 883, 721 23 V\"{a}ster\aa{}s, Sweden\\
Email: \{karl.lundengard, milica.rancic, sergei.silvestrov\}@mdh.se
\\ \\
$^{\dagger}$Department of Power Engineering, \\
Faculty of Electronic Engineering, University of Ni\v{s}\\ Aleksandra Medvedeva 14, 18000 Ni\v{s}, Serbia\\
Email: vesna.javor@elfak.ni.ac.rs}

\maketitle

\begin{abstract}
Multi-peaked analytically extended function (AEF), previously applied by the authors to modelling of lightning discharge currents, is used in this paper for representation of the electrostatic discharge (ESD) currents. The fitting to data is achieved by interpolation of certain data points. In order to minimize unstable behaviour, the exponents of the AEF are chosen from a certain arithmetic sequence and the interpolated points are chosen according to a D-optimal design. ESD currents' modelling is illustrated through two examples: one corresponding to an approximation of the IEC Standard 61000-4-2 waveshape, and the other to representation of some measured ESD current.
\end{abstract}

\section{Introduction}
Well-defined representation of real electrostatic discharge (ESD) currents is needed in order to establish realistic requirements for ESD generators used in testing of the equipment and devices, as well as to provide and improve the repeatability of tests. Such representations should be able to approximate the ESD currents waveshapes for various test levels, test set-ups and procedures, and also for various ESD conditions such as approach speeds, types of electrodes, relative arc length, humidity, etc. A mathematical function is needed for computer simulation of ESD phenomena, for verification of test generators and for improving standard waveshape definition. 

Functions previously proposed in the literature for modelling of ESD currents, are mostly linear combinations of exponential functions, Gaussian functions, Heidler functions or other functions, for a review see for example \cite{Lundengard_ICNPAA2016}.  
The Analytically Extended Function (AEF) has been previously proposed by the authors and successfully applied to lightning discharge modelling~\cite{Lundengard_PES,Lundengard_SJEE,Lundengard_MCAP} using least-square regression modelling. 

In this paper we analyse the applicability of the generalized multi-peaked AEF function to representation of ESD currents by interpolation of data points chosen according to a $D$-optimal design. This is illustrated through two examples corresponding to modelling of the IEC Standard 61000-4-2 waveshape,~\cite{IECStandard2000, IECStandard2009} and an experimentally measured ESD current from \cite{Katsivelis2010}.

\section{IEC 61000-4-2 Standard Current Waveshape}
ESD generators used in testing of the equipment and devices should be able to reproduce the same ESD current waveshape each time. This repeatability feature is ensured if the design is carried out in compliance with the requirements defined in the IEC 61000-4-2 Standard,~\cite{IECStandard2009}. 

Among other relevant issues, the Standard includes graphical representation of the typical ESD current, Fig.~\ref{fig:ESD_Standard}, and also defines, for a given test level voltage, required values of ESD current's key parameters. These are listed in Table~\ref{tab:ESDpar} for the case of the contact discharge, where:
\begin{itemize}
\item $I_{peak}$ is the ESD current initial peak;
\item $t_{r}$ is the rising time defined as the difference between time moments corresponding to $10\%$ and $90\%$ of the current peak $I_{peak}$, Fig.~\ref{fig:ESD_Standard};
\item $I_{30}$ and $I_{60}$ is the ESD current values calculated for time periods of 30 and 60 ns, respectively, starting from the time point corresponding to $10\%$ of $I_{peak}$, Fig.~\ref{fig:ESD_Standard}.
\end{itemize}  

\begin{table}[!t]
\renewcommand{\arraystretch}{1.5}
\caption{IEC 61000-4-2 Standard ESD Current and its Key Parameters,~\cite{IECStandard2009}.}
\label{tab:ESDpar}
\centering
\begin{tabular}{crcrc}
\hline
Voltage [kV] & $I_{peak}$ [A]& $t_{r}$ [ns] & $I_{30}$ [A]& $I_{60}$ [A] \\ 
\hline
2 & $7.5 \pm 15\%$ & $0.8 \pm 25\%$ & $4.0 \pm 30\%$ & $2.0 \pm 30\%$ \\
4 & $15.0 \pm 15\%$ & $0.8 \pm 25\%$ & $8.0 \pm 30\%$ & $4.0 \pm 30\%$ \\
6 & $22.5 \pm 15\%$ & $0.8 \pm 25\%$ & $12.0 \pm 30\%$ & $6.0 \pm 30\%$ \\
8 & $30.0 \pm 15\%$ & $0.8 \pm 25\%$ & $16.0 \pm 30\%$ & $8.0 \pm 30\%$ \\
\hline
\end{tabular}
\end{table}

\begin{figure}[!t]
 \centering
 \includegraphics[width=2.8in]{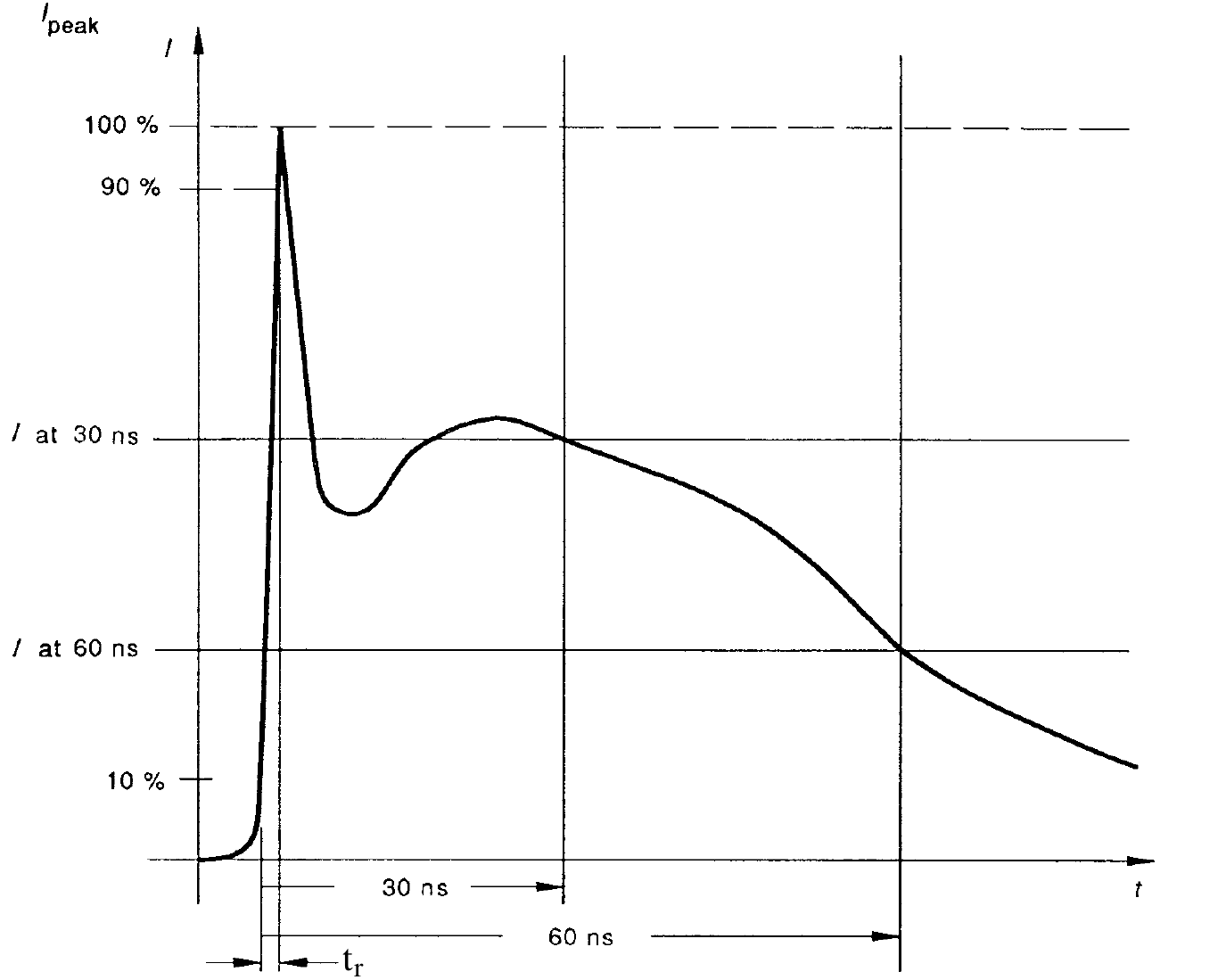}
 \caption{Illustration of the IEC 61000-4-2 Standard ESD current and its key parameters,~\cite{IECStandard2009}.}
 \label{fig:ESD_Standard}
\end{figure}

\section{Modelling of the IEC 61000-4-2 Standard Waveshape}
\subsection{Important Features of ESD Currents}
Various mathematical expressions have been introduced in the literature that can be used for representation of the ESD currents, either the IEC 61000-4-2 Standard one~\cite{IECStandard2009}, or experimentally measured ones, e.g.~\cite{Fotis2011}. These functions are to certain extent in accordance with the requirements given in Table~\ref{tab:ESDpar}. Furthermore, they have to satisfy the following:
\begin{itemize}
\item the value of the ESD current and its first derivative must be equal to zero at the moment $t=0$, since neither the transient current nor the radiated field generated by the ESD current can change abruptly at that moment.
\item the ESD current function must be time-integrable in order to allow numerical calculation of the ESD radiated fields.
\end{itemize} 

\subsection{Multi-Peaked Analytically Extended Function (AEF)}
A so-called multi-peaked analytically extended function (AEF) has been proposed and applied by the authors to lightning discharge current modelling in~\cite{Lundengard_SJEE,Lundengard_MCAP,Lundengard_PES}. Initial considerations on applying the function to ESD currents have also been made in \cite{Lundengard_ICNPAA2016}.

The AEF consists of scaled and translated functions of the form $x(\beta;t)=\left(te^{1-t}\right)^\beta$ that the authors have previously referred to as power-exponential functions~\cite{Lundengard_SJEE}.

Here we define the AEF with $p$ peaks as
\begin{equation}
 \label{eq:AEF_rise}
 i(t) = \displaystyle\sum_{k = 1}^{q-1} I_{m_k}+ I_{m_q} \sum_{k = 1}^{n_q} \eta_{q,k} x_{q,k}(t),
\end{equation}
for $t_{m_{q-1}} \leq t \leq t_{m_q}$, $1 \leq q \leq p$, and 
\begin{equation}
 \label{eq:AEF_decay}
 \displaystyle\sum_{k = 1}^{p} I_{m_k} \sum_{k = 1}^{n_{p+1}} \eta_{p+1,k} x_{p+1,k}(t),
\end{equation}
for $t_{m_p} \leq t$.

The current value of the first peak is denoted by $I_{m_1}$, the difference between each pair of subsequent peaks by $I_{m_2}, I_{m_3}, \ldots, I_{m_p}$, and their corresponding times by $t_{m_1}, t_{m_2}, \ldots, t_{m_p}$. In each time interval $q$, with $1 \leq q \leq p+1$, is denoted by $n_q$, $0 < n_q \in \mathbb{Z}$.  Parameters $\eta_{q,k}$ are such that $\eta_{q,k} \in \mathbb{R}$ for $q = 1,2,\ldots,p+1$, $k = 1,2,\ldots,n_q$ and $\displaystyle\sum_{k=1}^{n_q} \eta_{q,k} = 1$.
Furthermore $x_{q,k}(t),~1 \leq q \leq p+1$ is given by
\begin{equation}\label{eq:x_qk}
		x_{q,k}(t) =
		\begin{cases}
		 x\left(\beta_{q,k};\frac{t-t_{m_{q-1}}}{t_{m_q}-t_{m_{q-1}}}\right),~& 1 \leq q \leq p,\\
                  x\left(\beta_{q,k};\frac{t}{t_{m_q}}\right), & q=p+1.
        \end{cases} 
\end{equation}

\begin{remark}
 When previously applying the AEF, see \cite{Lundengard_PES,Lundengard_SJEE,Lundengard_MCAP}, all exponents ($\beta$-parameters) of the AEF were set to $\beta^2+1$ in order to guarantee that the derivative of the AEF is continuous. Here this condition will be satisfied in a different manner.
\end{remark}
	
Since the AEF is a linear function of elementary functions its derivative and integral can be found using standard methods. For explicit formulae please refer to~\cite[Theorems 1-3]{Lundengard_MCAP}.

Previously, the authors have fitted AEF functions to lightning discharge currents and ESD currents using the Marquardt least square method but have noticed that the obtained result varies greatly depending on how the waveforms are sampled. This is problematic, especially since the methodology becomes very computationally demanding when applied to large amounts of data. Here we will try one way to minimize the data needed but still enough to get an as good approximation as possible. 

The method examined here will be based on $D$-optimality of a regression model. A $D$-optimal design is found by choosing sample points such that the determinant of the Fischer information matrix of the model is minimized. For a standard linear regression model this is also equivalent, by the so-called Kiefer-Wolfowitz equivalence criterion, to $G$-optimality which means that the maximum of the prediction variance will be minimized. These are standard results in the theory of optimal experiment design and a summary can be found for example in \cite{Melas}.

Minimizing the prediction variance will in out case mean maximizing the robustness of the model. This does not guarantee a good approximation but it will increase the chances of the method working well when working with limited precision and noisy data and thus improve the chances of finding a good approximation when it is possible.

\section{$D$-Optimal Approximation for Exponents Given by a Class of Arithmetic Sequences}

It can be desirable to minimize the number of points used when constructing the approximation. One way to do this is to choose the $D$-optimal sampling points.

In this section we will only consider the case where in each interval the $n$ exponents, $\beta_1$, \ldots, $\beta_n$, are chosen according to
\[ \beta_m = \frac{k+m-1}{c} \]
where $k$ is a non-negative integer and $c$ a positive real number. Note that in order to guarantee continuity of the AEF derivative the condition is that $k > m$.

Then in each interval we want an approximation of the form
\[ y(t) = \sum_{i=1}^{n} \eta_i t^{\beta_i} e^{\beta_i(1-t)} \]
and by setting $z(t) = (t e^{1-t})^{\frac{1}{c}}$ then
\[ y(t) = \sum_{i=1}^{n} \eta_i z(t)^{k+i-1}. \]

If we have $n$ sample points, $t_i$, $i = 1,\ldots,n$, then the Fischer information matrix, $M$, of this system is $M = W^\top W$ where
\[ W = \begin{bmatrix}
           z(t_1)^k    &    z(t_2)^k    & \ldots &    z(t_n)^k    \\
         z(t_1)^{k+1}  &  z(t_2)^{k+1}  & \ldots &  z(t_n)^{k+1}  \\
            \vdots     &     \vdots     & \ddots &     \vdots     \\
        z(t_1)^{k+n-1} & z(t_2)^{k+n-1} & \ldots & z(t_n)^{k+n-1}
       \end{bmatrix}. \]
Thus if we want to maximize $\det(M) = \det(W)^2$ it is sufficient to maximize or minimize the determinant $\det(W)$. Set $z(t_i) = (t_i e^{1-t_i})^\frac{1}{c} = x_i$ then
\begin{gather} 
 w_n(t_1,\ldots,t_n) = \det(W) \nonumber \\
 \label{eq:fischerinfodet}
 = \left(\prod_{k=1}^{n} x_k\right) \left(\prod_{1\leq i < j \leq n} (x_j-x_i) \right).
\end{gather}

To find $t_i$ we will use the Lambert $W$ function. Formally the Lambert $W$ function is the function $W$ that satisfies $t = W(t e^t)$. Using $W$ we can invert $z(t)$ in the following way
\begin{align}
 t e^{1-t} = x^c & \Leftrightarrow -t e^{-t} = - e^{-1} x^c \nonumber \\
                 \label{eq:x_to_t}
                 & \Leftrightarrow t = -W(- e^{-1} x^c). 
\end{align}
The Lambert $W$ is multivalued but since we are only interested in real-valued solutions we are restricted to the main branches $W_0$ and $W_{-1}$. Since $W_0 \geq -1$ and $W_{-1} \leq -1$ the two branches correspond to the rising and decaying parts of the AEF respectively. We will deal with the details of finding the correct points for the two parts separately.

\subsection{$D$-Optimal Interpolation on the Rising Part}
Finding the $D$-optimal points on the rising part can be done using theorem \ref{thm:max_shift_det}.
\begin{theorem}
 \label{thm:max_shift_det}
 The determinant
 \[ w_n(k;x_1,\ldots,x_n) = \left(\prod_{i=1}^{n} x_i^k\right) \left(\prod_{1\leq i < j \leq n} (x_j-x_i) \right) \]
 is maximized or minimized on the cube $[0,1]^n$ when $x_1 < \ldots < x_{n-1}$ are roots of the Jacobi polynomial
 \begin{equation*}
  \label{eq:doptimalpolynomial}
  P_{n-1}^{(2k-1,0)}(1-2x) = \frac{(2k)^{\overline{n-1}}}{(n-1)!} \sum_{i=0}^{n-1} (-1)^n \binom{n-1}{i} \frac{(2k+n)^{\overline{i}}}{(2k)^{\overline{i}}} x^i
 \end{equation*}
 and $x_n = 1$, or some permutation thereof. \\
 Here $a^{\overline{b}}$ is the rising factorial $a^{\overline{b}} = a(a+1)\cdots(a+b-1)$.
\end{theorem}

\begin{proof}
 Without loss of generality we can assume $0 < x_1 < x_2 < \ldots < x_{n-1} < x_n \leq 1$. Fix all $x_i$ except $x_n$. When $x_n$ increases all factors of $w_n$ that contain $x_n$ will also increase, thus $w_n$ will reach its maximum value on the edge of the cube where $x_n = 1$.
 Using the method of Lagrange multipliers in the plane given by $x_n = 1$ gives
 \[\frac{\partial w_n}{\partial x_j}  = w_n(k;x_1,\ldots,x_n) \left(\frac{k}{x_j} + \sum_{\genfrac{}{}{0pt}{}{i=1}{i \neq j}}^{n} \frac{1}{x_j-x_i} \right) = 0, \]
 for $j=1,\ldots,n-1$. By setting $f(x) = \displaystyle\prod_{i=1}^{n} (x-x_i)$ we get
 \begin{gather}
  \frac{k}{x_j} + \sum_{\genfrac{}{}{0pt}{}{i=1}{i \neq j}}^{n} \frac{1}{x_j-x_i} = 0 \Leftrightarrow \frac{k}{x_j} + \frac{1}{2} \frac{f''(x_j)}{f'(x_j)} = 0 \nonumber \\
  \label{eq:diff_eq_1}
  \Leftrightarrow x_j f''(x_j) + 2kf'(x_j) = 0
 \end{gather}
 for $j=1,\ldots,n-1$.
 Since $f(x)$ is a polynomial of degree $n$ that has $x = 1$ as a root then equation (\ref{eq:diff_eq_1}) implies
 \begin{equation*}
  x f''(x) + 2kf'(x) = c \frac{f(x)}{x-1}
 \end{equation*}
 where $c$ is some constant.
 Set $f(x) = (x-1)g(x)$ and the resulting differential equation is
 \begin{equation*}
  x(x-1)g''(x) + ((2k+2)x-2k)g'(x) + (2k-c)g(x) = 0.
 \end{equation*}
The constant $c$ can be found by examining the terms with degree $n-1$ and is given by $c = 2k + (n-1)(2k+n)$, thus
 \begin{gather}
   x(1-x)g''(x) + (2k-(2k+2)x)g'(x) \nonumber \\ 
   \label{eq:diff_eq_2}
   + (n-1)(2k+n)g(x) = 0.
 \end{gather}
 Comparing (\ref{eq:diff_eq_2}) with the standard form of the hypergeometric function \cite{AbramowitzStegun}
 \[ x(1-x) g''(x) + (c-(a+b+1)x) g'(x) - ab g(x) = 0 \]
 shows that $g(x)$ can be expressed as follows
 \begin{align*} 
  g(x) & = C \cdot \,_2\!F_1(1-n,2k+n;2k,x) \\
       & = C \cdot \frac{(2k)^{\overline{n-1}}}{(n-1)!} \sum_{i=0}^{n-1} (-1)^i \binom{n-1}{i} \frac{(2k+n)^{\overline{i}}}{(2k)^{\overline{i}}} x^i
 \end{align*}
 where $C$ is an arbitrary constant and since we are only interested in the roots of the polynomial we can chose $C$ so that it gives the desired form of the expression.
 The connection to the Jacobi polynomial is given by \cite{AbramowitzStegun}
 \[ _2F_1(-m,1+\alpha+\beta+n;\alpha+1;x)=\frac{m!}{(\alpha+1)^{\overline{m}}} P^{(\alpha,\beta)}_m(1-2x), \]
 and $\alpha = 2k-1$, $\beta = 0$, $m = n-1$ gives the expression in theorem \ref{thm:max_shift_det}.
\end{proof}

We can now find the $D$-optimal $t$-values using the upper branch of the Lambert W function as described in equation~(\ref{eq:x_to_t}),
\[ t_i = -W_0(- e^{-1} x_i^m), \] 
where $x_i$ are the roots of the Jacobi polynomial given in theorem \ref{thm:max_shift_det}. Since $-1 \leq W_0(x) \leq 0$ for $-e^{-1} \leq 0$ this will always give $0 \leq t_i \leq 1$.

\begin{remark}
 Note that $x_n = 1$ means that $t_n = t_q$ and also is equivalent to the condition that $\displaystyle\sum_{k=1}^{n_q} \eta_{q,k} = 1$. In other words we are interpolating the peak and $p-1$ points inside each interval.
\end{remark}

\subsection{$D$-Optimal Interpolation on the Decaying Part}
Finding the $D$-optimal points for the decaying part is more difficult than it is for the rising part. Suppose we denote the largest value for time that can reasonably be used (for computational or experimental reasons) with $t_{max}$. This corresponds to some value $x_{max} = (t_{max} \exp(1-t_{max}))^{\frac{1}{c}}$. Ideally we would want a corresponding theorem to theorem \ref{thm:max_shift_det} over $[1,x_{max}]^n$ instead of $[0,1]^n$. It is easy to see that if $x_i = 0$ or $x_i = 1$ for some $1 \leq x_i \leq n-1$ then $w_n(k;x_1,\ldots,x_n) = 0$ and thus there must exist some local extreme point such that $0 < x_1 < x_2 < \ldots < x_{n-1} < 1$. This is no longer guaranteed when considering the volume $[1,x_{max}]^n$ instead. Therefore we will instead extend theorem \ref{thm:max_shift_det} to the volume $[0,x_{max}]^n$ and give an extra constraint on the parameter $k$ that guarantees $1 < x_1 < x_2 < \ldots < x_{n-1} < x_n = x_{max}$.

\begin{theorem}
 Let $y_1 < y_2 < \ldots < y_{n-1}$ be the roots of the Jacobi polynomial $P_{n-1}^{(2k-1,0)}(1-2y)$. If $k$ is chosen such that $1 < x_{max} \cdot y_1$ then the determinant $w_n(k;x_1,\ldots,x_n)$ given in Theorem \ref{thm:max_shift_det} is maximized or minimized on the cube $[1,x_{max}]^n$ (where $x_{max} > 1$) when $x_i = x_{max} \cdot y_i$ and $x_n = x_{max}$, or some permutation thereof.
\end{theorem}
\begin{proof}
 This theorem follows from Theorem \ref{thm:max_shift_det} combined with the fact that \\ $w_n(k;x_1,\ldots,x_n)$ is a homogeneous polynomial. 
 Since 
 \[ w_n(k; b \cdot x_1,\ldots, c \cdot x_n) = b^{k+\frac{n(n-1)}{2}} \cdot w_n(k;x_1,\ldots,x_n) \]
 if $(x_1,\ldots,x_n)$ is an extreme point in $[0,1]^n$ then $(b \cdot x_1,\ldots,b \cdot x_n)$ is an extreme point in $[0,b]^n$. Thus by theorem \ref{thm:max_shift_det} the points given by $x_i = x_{max} \cdot y_i$ will maximize or minimize $w_n(k;x_1,\ldots,x_n)$ on $[0,b]^n$.
\end{proof}
\begin{remark}
 It is in many cases possible to ensure the condition $1 < x_{max} \cdot y_1$ without actually calculating the roots of $P_{n-1}^{(2k-1,0)}(1-2y)$. In the literature on orthogonal polynomials there are many expressions for upper and lower bounds of the roots of the Jacobi polynomials. For instance in \cite{DriverJordaan} an upper bound on the largest root of a Jacobi polynomial is given that in our case can be rewritten as
 \[ y_1 > 1-\frac{3}{4k^2+2kn+n^2-k-2n+1} \]
 and thus
 \[ 1-\frac{3}{4k^2+2kn+n^2-k-2n+1} > \frac{1}{x_{max}} \]
 guarantees that $1 < x_{max} \cdot y_1$. If a more precise condition is needed there are expressions that give tighter bounds of the largest root of the Jacobi polynomials, see \cite{Lucas}. 
\end{remark}
We can now find the $D$-optimal $t$-values using the lower branch of the Lambert W function as described in equation (\ref{eq:x_to_t}),
\[ t_i = -W_{-1}(- e^{-1} x_i^c), \] 
where $x_i$ are the roots of the Jacobi polynomial given in Theorem \ref{thm:max_shift_det}. Since $-1 \leq W_{-1}(x) < -\infty$ for $-e^{-1} \leq x \leq 0$ this will always give $1 \leq t_i < b$.

\begin{remark}
 Note that here just like in the rising part $t_n = t_p$ which means that we will interpolate to the final peak as well as $p-1$ points in the decaying part.
\end{remark}

\section{Examples of Models from Applications and Experiments}

In this section some results of applying the described scheme to two different waveforms will be presented. The two waveforms are the Standard ESD current given in IEC 61000-4-2,~\cite{IECStandard2009} and a waveform from experimental measurements from ~\cite{Katsivelis2010}.

The values of $n$, $k$ and $c$ have been chosen by manual experimentation and since both waveforms are given as data rather than explicit functions the $D$-optimal points have been calculated and then the closest available data points have been chosen.

Note that the quality of the results can vary greatly depending on how the $k$ and $m$ parameters are chosen before this type of approximation scheme is applied, and in practice a strategy for choosing the values effectively should be devised. In many cases increasing the number of interpolation points, $n$, improves the results but there are many cases where the interpolation is not stable.

\subsection{Interpolated AEF Representing the IEC 61000-4-2 Standard Current}

In this section we present the results of fitting 2- and 3-peak AEF to the Standard ESD current given in IEC 61000-4-2. Data points which are used in the optimization procedure are manually sampled from the graphically given Standard \cite{IECStandard2009} current function. The peak currents and corresponding times are also extracted, and the results of $D$-optimal interpolation with 2 and 3 peaks are illustrated, see Fig.~\ref{fig:ESD_AEF_max} and \ref{fig:ESD_AEF}. The parameters are listed in Table~\ref{tab:Standard}. In the illustrated examples a fairly good fit is found but typically areas with steep rise and the decay part are somewhat more difficult to fit with good accuracy than the other parts of the waveform.
\begin{figure}[!t]
 \centering
 \includegraphics[width=2.5in]{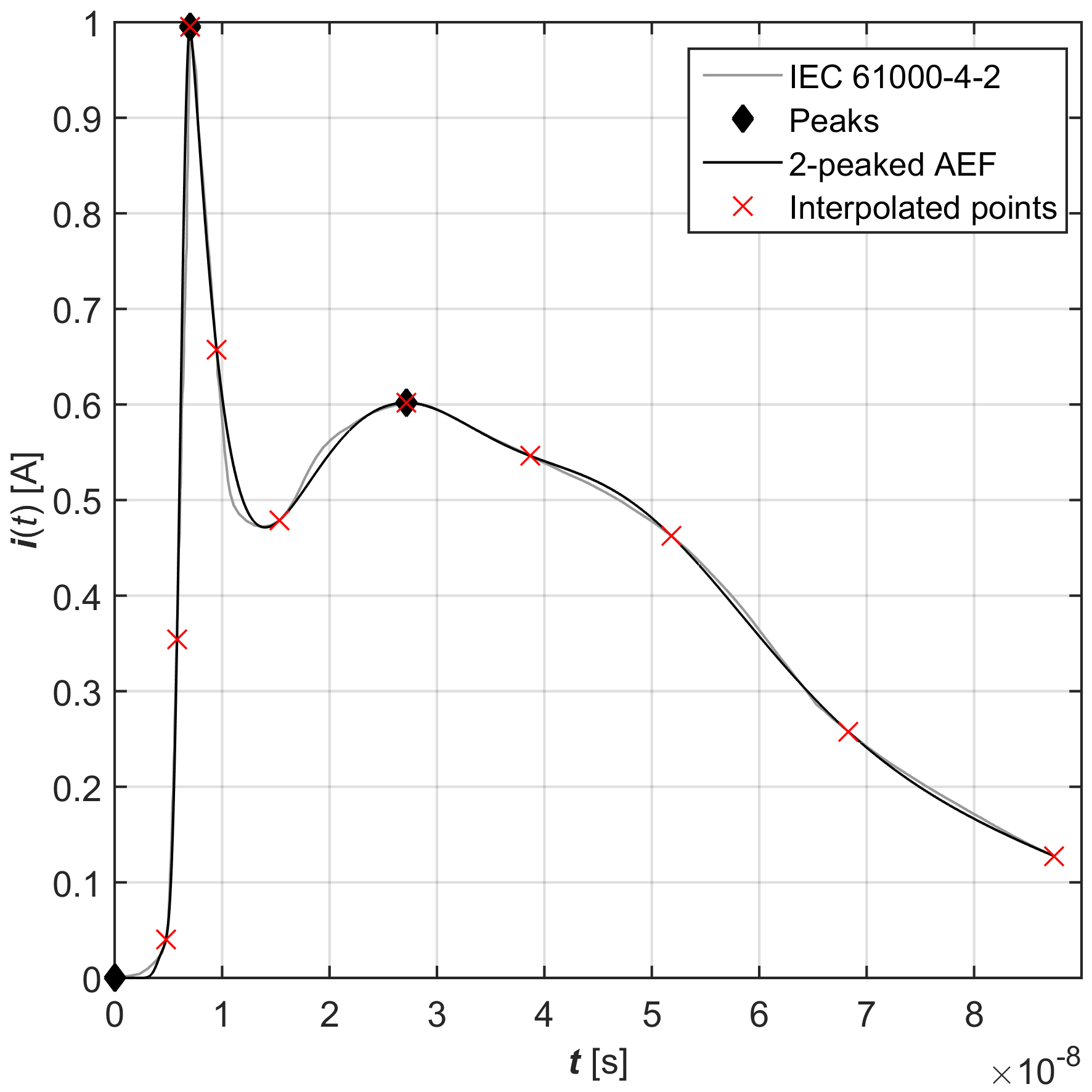}
 \caption{2-peaked AEF representing the IEC 61000-4-2 Standard ESD current waveshape for 4kV. Parameters are given in Table \ref{tab:Standard}.}
 \label{fig:ESD_AEF_max}
\end{figure}
\begin{figure}[!t]
 \centering
 \includegraphics[width=2.5in]{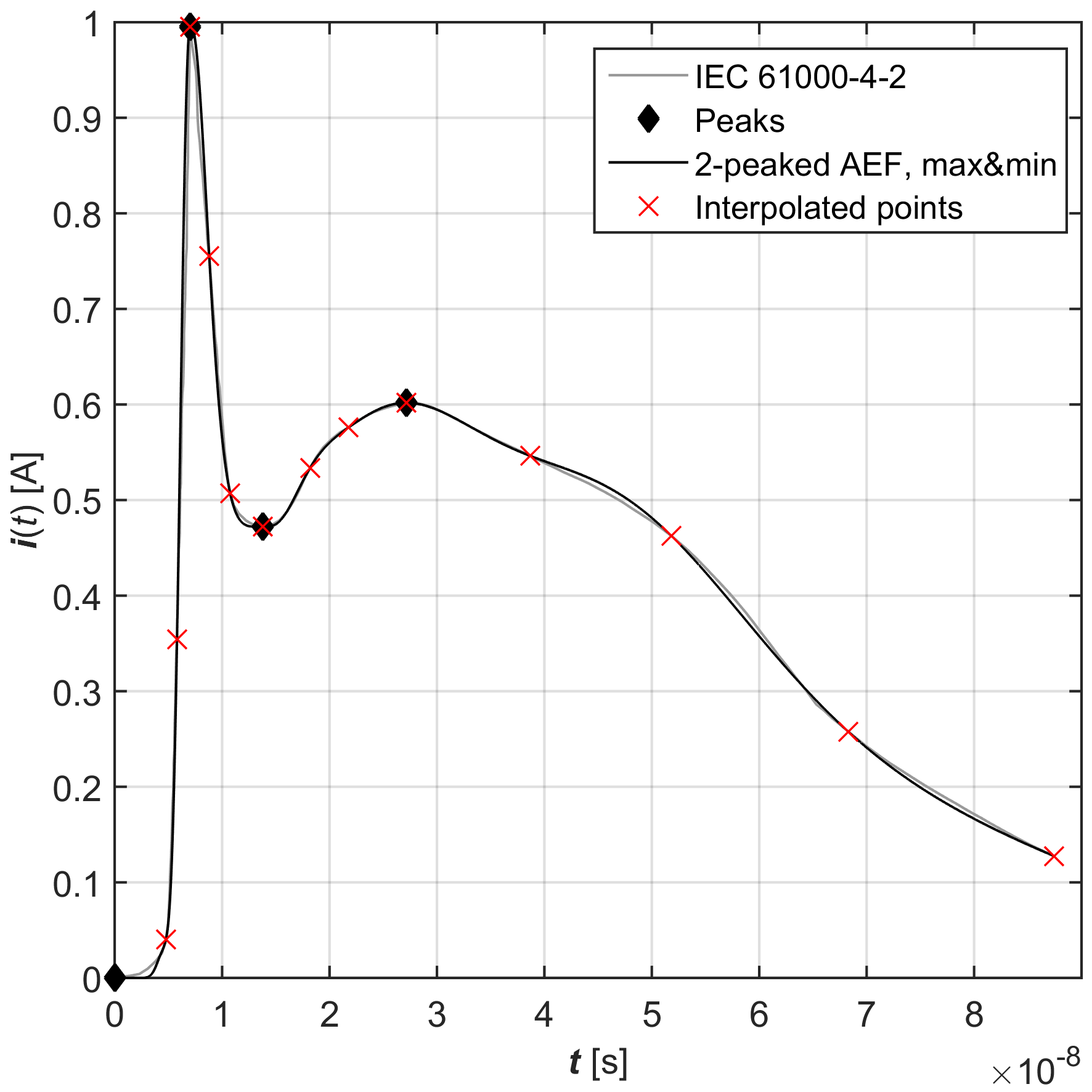}
 \caption{3-peaked AEF representing the IEC 61000-4-2 Standard ESD current waveshape for 4kV. Parameters are given in Table \ref{tab:Standard}.}
 \label{fig:ESD_AEF}
\end{figure}

\begin{table}[!t]
 \renewcommand{\arraystretch}{1.5}
 \caption{Parameters' Values of the Multi-Peaked AEFs Representing the IEC 61000-4-2 Standard Waveshape.}
 \label{tab:Standard}
 \centering
 \begin{tabular}{cccccc}
  \hline
  \multicolumn{6}{c}{Local maxima and minima and corresponding times}\\
  \multicolumn{6}{c}{extracted from the IEC 61000-4-2,~\cite{IECStandard2009}}\\ 
  \hline
  \multicolumn{2}{c}{$~~~~I_{max1}$ = 15 [A] } & \multicolumn{2}{c}{$~~I_{min1}$  = 7.1484 [A] } & \multicolumn{2}{c}{$I_{max2}$ = 9.0921 [A]}\\ 
  \multicolumn{2}{c}{$~~~~t_{max1}$ = 6.89 [ns]} &\multicolumn{2}{c}{$~~t_{min1}$ = 12.85 [ns]} &\multicolumn{2}{c}{ $t_{max2}$ = 25.54 [ns]}\\
  \hline
  \multicolumn{6}{c}{Parameters of interpolated AEF shown in Fig.~\ref{fig:ESD_AEF_max}} \\ \hline
  \multicolumn{3}{c}{Interval} & \multicolumn{1}{c}{$n$} & \multicolumn{1}{c}{$k$}& \multicolumn{1}{c}{$c$} \\
  \hline
  \multicolumn{3}{c}{$0 \leq t \leq t_{max1}$}        & \multicolumn{1}{c}{$3$} & \multicolumn{1}{c}{$35$} & \multicolumn{1}{c}{$1$} \\
  \multicolumn{3}{c}{$t_{max1} \leq t \leq t_{max2}$} & \multicolumn{1}{c}{$3$} & \multicolumn{1}{c}{$3$}  & \multicolumn{1}{c}{$2$} \\
  \multicolumn{3}{c}{$t_{max2} < t$}                  & \multicolumn{1}{c}{$5$} & \multicolumn{1}{c}{$3$}  & \multicolumn{1}{c}{$1$} \\
  \hline
  \multicolumn{6}{c}{Parameters of interpolated AEF shown in Fig.~\ref{fig:ESD_AEF}} \\ \hline
  \multicolumn{3}{c}{Interval} & \multicolumn{1}{c}{$n$} & \multicolumn{1}{c}{$k$}& \multicolumn{1}{c}{$c$} \\
  \hline
  \multicolumn{3}{c}{$0 \leq t \leq t_{max1}$}        & \multicolumn{1}{c}{$3$} & \multicolumn{1}{c}{$35$} & \multicolumn{1}{c}{$1$} \\
  \multicolumn{3}{c}{$t_{max1} \leq t \leq t_{min1}$} & \multicolumn{1}{c}{$3$} & \multicolumn{1}{c}{$3$}  & \multicolumn{1}{c}{$1$} \\
  \multicolumn{3}{c}{$t_{min1} \leq t \leq t_{max2}$} & \multicolumn{1}{c}{$3$} & \multicolumn{1}{c}{$4$}  & \multicolumn{1}{c}{$1$} \\
  \multicolumn{3}{c}{$t_{max2} < t$}                  & \multicolumn{1}{c}{$5$} & \multicolumn{1}{c}{$3$}  & \multicolumn{1}{c}{$1$} \\
  \hline
 \end{tabular}
\end{table}

\subsection{3-peaked AEF Representing Measured Data}
In this section we present the results of fitting a 1-, 2- and a 3-peaked AEF to a waveform from experimental measurements from ~\cite{Katsivelis2010}. The result is also compared to a common type of function used for modelling ESD current, also from \cite{Katsivelis2010}.

In Figs. \ref{fig:Kats_1}-\ref{fig:Kats_3} the results of the interpolation of $D$-optimal points for certain parameters are shown together with the measured data, as well as a sum of two Heidler functions that was fitted to the experimental data in \cite{Katsivelis2010}. This function is given by
\begin{equation} 
i(t) = I_1 \frac{\left(\frac{t}{\tau_1}\right)^{n_H}}{1+\left(\frac{t}{\tau_1}\right)^{n_H}} e^{-\frac{t}{\tau_2}} + I_2 \frac{\left(\frac{t}{\tau_3}\right)^{n_H}}{1+\left(\frac{t}{\tau_3}\right)^{n_H}} e^{-\frac{t}{\tau_4}}, 
\end{equation}
with
\begin{gather*} 
I_1 = 31.365 \mathrm{~A},~I_2 = 6.854 \mathrm{~A},~n_H=4.036, \\
\tau_1=1.226 \mathrm{~ns},~\tau_2=1.359 \mathrm{~ns}, \\
\tau_3=3,982 \mathrm{~ns},~\tau_4=28.817 \mathrm{~ns}. 
\end{gather*}

Note that this function does not reproduce the second local minimum but that all three AEF functions can reproduce all local minima and maxima (to a modest degree of accuracy) when suitable values for the $n$, $k$ and $m$ parameters are chosen.

\begin{table}[!t]
 \renewcommand{\arraystretch}{1.5}
 \caption{Parameters' Values of Multi-Peaked AEFs Representing Experimental Data.}
 \label{tab:Katsivelis}
 \centering
 \begin{tabular}{cccccc}
  \hline
  \multicolumn{6}{c}{Local maxima and corresponding times extracted from~\cite[Fig.~3]{Katsivelis2010}}\\ 
  \hline
  \multicolumn{2}{c}{$~~~~I_{max1}$ = 7.37 [A] } & \multicolumn{2}{c}{$~~I_{max2}$  =  5.02 [A] } & \multicolumn{2}{c}{$I_{max3}$ = 3.82 [A]}\\ 
  \multicolumn{2}{c}{$~~~~t_{max1}$ = 1.23 [ns]} & \multicolumn{2}{c}{$~~t_{max2}$ = 6.39 [ns]}   & \multicolumn{2}{c}{ $t_{max3}$ = 15.5 [ns]}\\
  \hline
  \multicolumn{6}{c}{Parameters of interpolated AEF shown in Fig.~\ref{fig:Kats_1}} \\ \hline
  \multicolumn{3}{c}{Interval} & \multicolumn{1}{c}{$n$} & \multicolumn{1}{c}{$k$}& \multicolumn{1}{c}{$c$} \\
  \hline
  \multicolumn{3}{c}{$0 \leq t \leq t_{max3}$} & \multicolumn{1}{c}{$12$} & \multicolumn{1}{c}{$8$} & \multicolumn{1}{c}{$0.8$} \\
  \multicolumn{3}{c}{$t_{max3} < t$}           & \multicolumn{1}{c}{$6$}  & \multicolumn{1}{c}{$2$} & \multicolumn{1}{c}{$1$} \\
  \hline
  \multicolumn{6}{c}{Parameters of interpolated AEF shown in Fig.~\ref{fig:Kats_2}} \\ \hline
  \multicolumn{3}{c}{Interval} & \multicolumn{1}{c}{$n$} & \multicolumn{1}{c}{$k$}& \multicolumn{1}{c}{$c$} \\
  \hline
  \multicolumn{3}{c}{$0 \leq t \leq t_{max1}$}        & \multicolumn{1}{c}{$5$} & \multicolumn{1}{c}{$40$} & \multicolumn{1}{c}{$1$} \\
  \multicolumn{3}{c}{$t_{max1} \leq t \leq t_{max3}$} & \multicolumn{1}{c}{$7$} & \multicolumn{1}{c}{$4$}  & \multicolumn{1}{c}{$1$} \\
  \multicolumn{3}{c}{$t_{max3} < t$}                  & \multicolumn{1}{c}{$6$} & \multicolumn{1}{c}{$2$}  & \multicolumn{1}{c}{$1$} \\
  \hline
  \multicolumn{6}{c}{Parameters of interpolated AEF shown in Fig.~\ref{fig:Kats_3}} \\ \hline
  \multicolumn{3}{c}{Interval} & \multicolumn{1}{c}{$n$} & \multicolumn{1}{c}{$k$}& \multicolumn{1}{c}{$c$} \\
  \hline
  \multicolumn{3}{c}{$0 \leq t \leq t_{max1}$}        & \multicolumn{1}{c}{$5$} & \multicolumn{1}{c}{$40$} & \multicolumn{1}{c}{$1$} \\
  \multicolumn{3}{c}{$t_{max1} \leq t \leq t_{max2}$} & \multicolumn{1}{c}{$3$} & \multicolumn{1}{c}{$3$}  & \multicolumn{1}{c}{$1$} \\
  \multicolumn{3}{c}{$t_{max2} \leq t \leq t_{max3}$} & \multicolumn{1}{c}{$4$} & \multicolumn{1}{c}{$4$}  & \multicolumn{1}{c}{$1$} \\
  \multicolumn{3}{c}{$t_{max3} < t$}                  & \multicolumn{1}{c}{$6$} & \multicolumn{1}{c}{$2$}  & \multicolumn{1}{c}{$1$} \\
  \hline
 \end{tabular}
\end{table}

\begin{figure}[!t]
 \centering
 \includegraphics[width=2.5in]{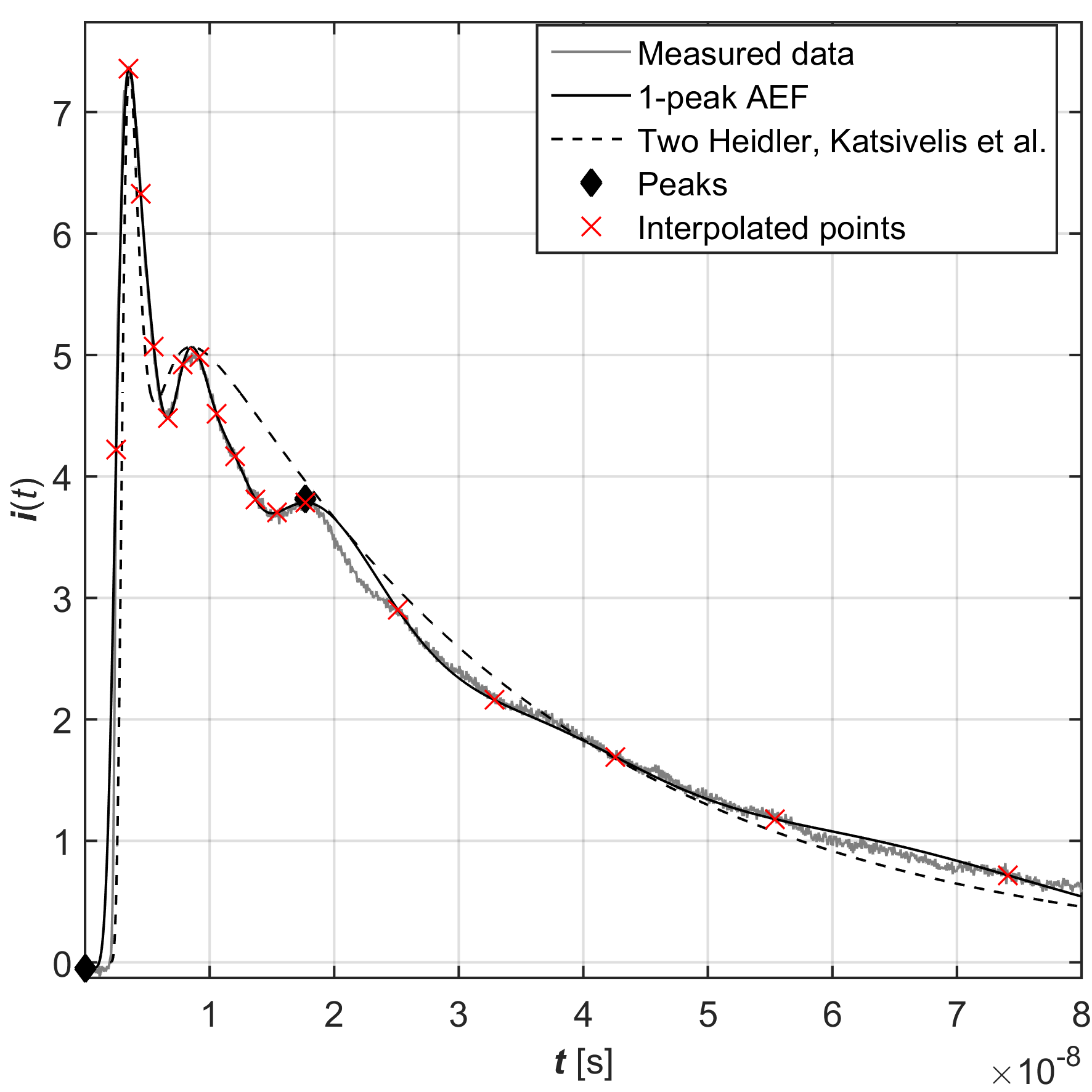}
 \caption{1-peak AEF interpolated to $D$-optimal points chosen from measured ESD current from~\cite[Fig.3]{Katsivelis2010}. Parameters are given in Table \ref{tab:Katsivelis}.}
 \label{fig:Kats_1}
\end{figure}
\begin{figure}[!t]
 \centering
 \includegraphics[width=2.5in]{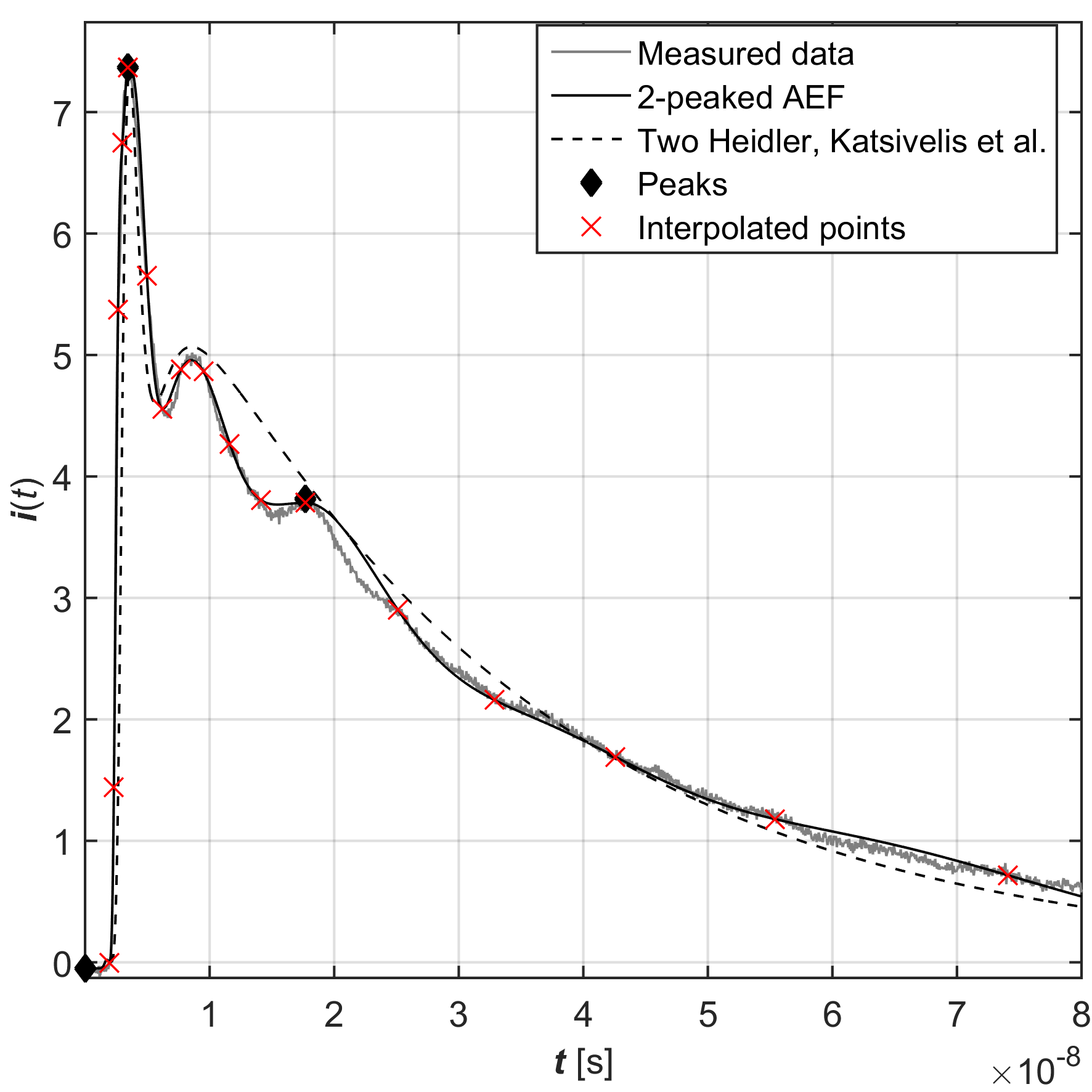}
 \caption{2-peaked AEF interpolated to $D$-optimal points chosen from measured ESD current from~\cite[Fig.3]{Katsivelis2010}. Parameters are givend in Table \ref{tab:Katsivelis}.}
 \label{fig:Kats_2}
\end{figure}
\begin{figure}[!t]
 \centering
 \includegraphics[width=2.5in]{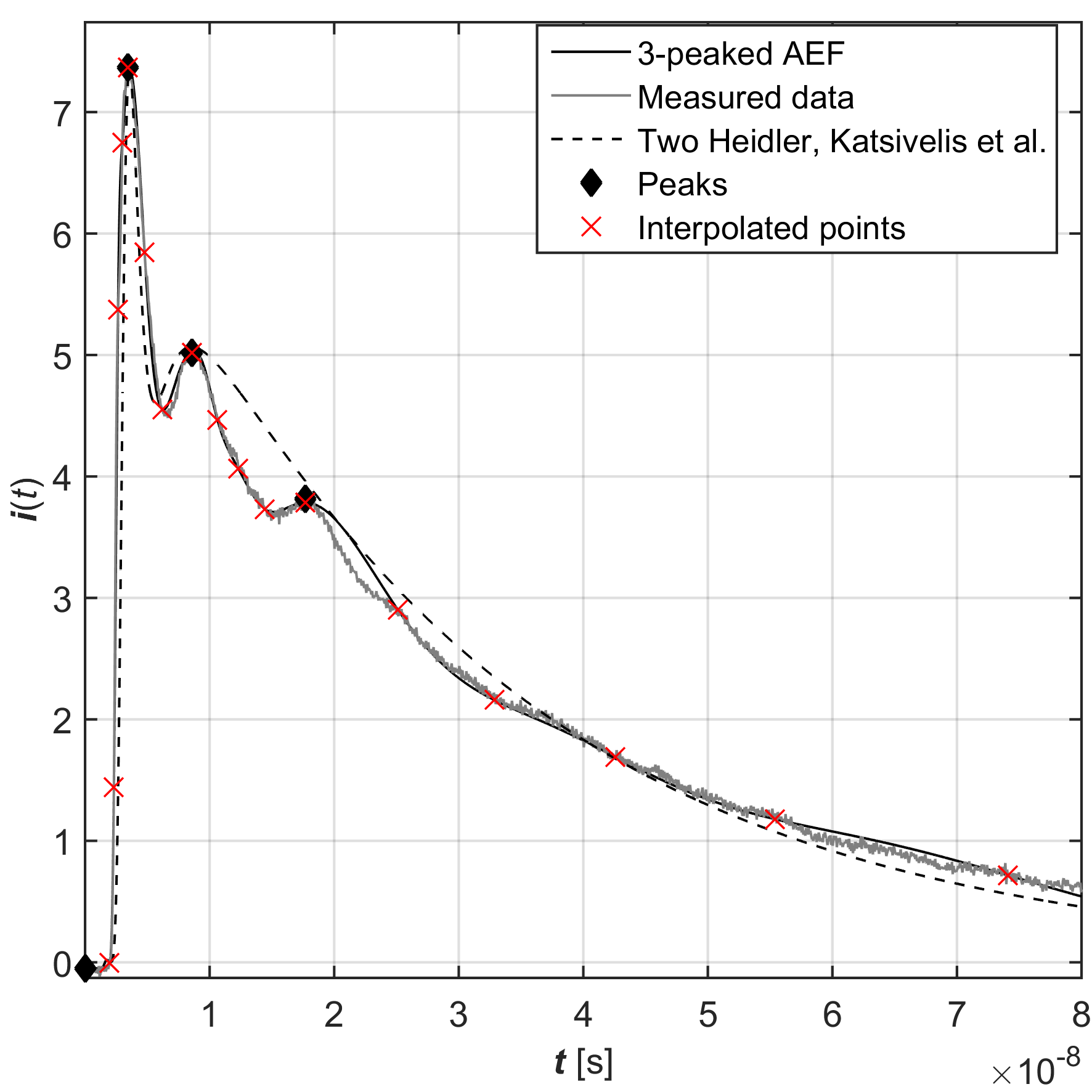}
 \caption{3-peaked AEF interpolated to $D$-optimal points chosen from measured ESD current from~\cite[Fig.3]{Katsivelis2010}. Parameters are given in Table \ref{tab:Katsivelis}.}
 \label{fig:Kats_3}
\end{figure}

\section{Conclusion}
In this work we examine a mathematical model for representation of ESD currents, either the IEC 61000-4-2 Standard one,~\cite{IECStandard2009}, or experimentally measured ones. The model has been proposed and successfully applied to lightning current modelling in~\cite{Lundengard_PES,Lundengard_SJEE,Lundengard_MCAP} and named the multi-peaked analytically extended function (AEF).

It conforms to the requirements for the ESD current and its first derivative, which are imposed by the Standard \cite{IECStandard2009} stating that they must be equal to zero at moment $t=0$. Furthermore, the AEF function is time-integrable,~\cite{Lundengard_MCAP}, which is necessary for numerical calculation of radiated fields originating from the ESD current. 

Here we consider how the model can be fitted to a waveform using $D$-optimal interpolation and the resulting methodology is illustrated on the IEC 61000-4-2 Standard waveform ~\cite{IECStandard2009} and experimental data from \cite{Katsivelis2010}.

The resulting methodology can give fairly accurate results even with a modest number of interpolated points but strategies for choosing some of the involved parameters should be further investigated.

\section*{Acknowledgement}

The authors would like to thank Dr. Pavlos Katsivelis from the High Voltage Laboratory, School of Electrical and Computer Engineering, National Technical University of Athens, Greece, for providing measured ESD current data.

\end{document}